\documentclass[a4paper]{llncs}
\usepackage{amsmath}
\usepackage{amsfonts}
\usepackage{amssymb}
\usepackage{graphicx}
\usepackage{color}
\usepackage{textcomp} 
\usepackage{nicefrac}
\usepackage{wrapfig}
\usepackage{multirow}
\usepackage{bigstrut}
\usepackage{longtable}
\usepackage[colorlinks=true, urlcolor=cyan]{hyperref}
\usepackage{leftidx}
\newcommand{\remove}[1] {}
\newtheorem{fact}{\bf Fact}
\renewcommand{\epsilon}{\varepsilon}

\newlength{\pgmtab}
\newcommand\blfootnote[1]{%
  \begingroup
  \renewcommand\thefootnote{}\footnote{#1}%
  \addtocounter{footnote}{-1}%
  \endgroup
}

\usepackage{leftidx}
\usepackage{graphicx}
\usepackage[labelformat=empty]{subfig}
\begin{document}

\title{Optimizing the Social Cost  of  Congestion Games  \\ by Imposing Variable Delays}

\author{Josep D{\'i}az\inst{1} \and Ioannis Giotis\inst{1,2} \and Lefteris Kirousis\inst{3} \and Yiannis Mourtos\inst{2} \and\\ Maria J. Serna\inst{1}}

\institute{Departament de Llenguatges i Sistemes Inform\`{a}tics\\ Universitat Polit\`{e}cnica de Catalunya, Barcelona
\and 
Department of Management Science and Technology \\ Athens University of Economics and Business, Greece
\and 
Department of Mathematics,  National \& Kapodistrian University of Athens, Greece 
\\ and Computer Technology Institute \& Press  ``Diophantus'', Patras, Greece}

\maketitle
%%%%%%%%%%%%%%%%%%%%%%%%%%%%%%%%%%%%%%%%%%%%%%%%%%%%%%%%%%%%%%%%%%%%%%%%%%%%%%%%%%%%%%%%
\begin{abstract} We describe a new coordination mechanism for non-atomic congestion games that leads to a  (selfish) social cost  which is arbitrarily close to the non-selfish optimal. This mechanism does not incur  any additional extra cost, like tolls,  which are usually differentiated from  the social cost as expressed in terms of delays only.  

\protect\blfootnote{\scriptsize Contact: \email{\{diaz,igiotis,mjserna\}@lsi.upc.edu,lkirousis@math.uoa.gr,mourtos@aueb.gr}}\blfootnote{\scriptsize Josep D{\'i}az, Maria J. Serna and Ioannis Giotis supported by the CICYT project TIN-2007-66523 (FORMALISM). This research has been co-financed by the European Union (European Social Fund – ESF) and Greek national funds through the Operational Program ``Education and Lifelong Learning" of the National Strategic Reference Framework (NSRF) - Research Funding Program: Thales. Investing in knowledge society through the European Social Fund.}
\end{abstract}

\section{Introduction}
Selfish behavior is one of the primary reasons many systems with multiple agents deviate from desirable outcomes. Allowing players to prioritize solely their own benefit can lead to social inefficiency, even in outcomes where no one is better off, compared to an optimal solution yet the social welfare is greatly diminished.

This type of behavior has been analyzed in various contexts and has often been verified in practice. A key such area is transportation and network routing where selfish selection among possible routes can lead to congestion with accompanying economical and environmental issues.

There exists a great amount of literature trying to address this problem both on the basis of theoretical analysis and in practice. A standard and reasonable modeling assumption is that all users choosing to use a particular road experience the same amount of latency. Then, the outcome is typically evaluated against the optimal social welfare outcome which minimizes the average latency the users experience.

Various approaches have been proposed to steer the selfishly constructed outcome towards optimal social welfare. The main idea is usually to incentivize the users to alter their selections to better ones, typically through the use of tolls or similar measures.

We propose an alternative approach that alters the way users experience latency and can offer significant improvements on average latency. Instead of all users experiencing the same  latency, as if everyone rushing to use the road at the same time, we propose to implement variable  latencies through a prioritization scheme. Some users will experience smaller latency than before, while others will experience longer latencies. We show that our system, if users behave selfishly as expected, achieves approximately the optimal social welfare. 

It is important to note that our system's average latency on each road closely matches the road's average latency without the system in place, hence our system falls under the notion of coordination mechanisms and we have not ``cheated'' by decreasing  latencies or imposing tolls; we are simply distributing the resource differently. This holds on any instance, not just equilibrium settings, which means that even non stable situations we do not get worse performance. Furthermore, we do not need to know the \emph{demand} in advance, i.e. our system delivers the social optimum for all possible total amounts of traffic. Our only requirement is that the latency induced on each road is a non-negative, non-decreasing, continuously differentiable and convex function of the traffic.

We believe our system has basis for practical implementation. For example, some countries have already implemented metered highway entrance ramps which can vary the latency of incoming drivers. Traffic signals can also be used in an urban environment to implement this aspect of our mechanism. We deliberately leave the prioritization scheme generic to allow for different such approaches with our only requirement being that users choosing to alter their current selection are forced to experience maximal  latency in their new selection, a reasonable requirement as typically someone that alters their selection in a running system ends up at the end of a queue.

We examine our system in the generic scheme of congestion games. As such, it can have other application besides traffic routing. One interesting application could be in the context of job scheduling. Again, in a typical model, each user choosing to use a particular resource experiences the same latency,  for example computing jobs running in parallel on a computer. By prioritizing jobs according to our proposed mechanism such that some jobs complete faster and some slower than before we can achieve optimal average job completion times under selfish behavior. We note that this can be easily implemented by an administrator using system priorities. 

\section{Related work}
The fact that selfish behavior can lead to inefficiency has been long studied in the context of transportation theory~\cite{pigou,wardrop1952road}. More recently, Koutsoupias and Papadimitriou introduced the \emph{Price of Anarchy} as a measure of this inefficiency~\cite{DBLP:journals/csr/KoutsoupiasP09,DBLP:conf/stoc/Papadimitriou01}. Exploration of this metric in the context of selfish routing was then greatly progressed by Roughgarden and Tardos~\cite{DBLP:journals/jacm/RoughgardenT02,DBLP:books/daglib/0012676} which bounded the price of anarchy for different classes of latency functions.

Naturally, ways to improve inefficient outcomes were investigated, with a prime example being the imposement of \emph{tolls}~\cite{DBLP:conf/stoc/ColeDR03,fleischer2004tolls,karakostas2004edge}. While this approach achieves optimal social welfare regarding latencies, it introduces a cost separation to the players as the tolls' cost is affecting behavior but is not accounted for in the objective function.

Coordination mechanisms were introduced by Christodoulou et al.~\cite{DBLP:journals/tcs/ChristodoulouKN09} as a way to ``shape'' latency functions to steer the selfishly dictated outcome towards greater social welfare. There are two main  restrictions in a coordination mechanism as defined in \cite{DBLP:journals/tcs/ChristodoulouKN09}, first the latency on each resource is not reduced and secondly the benchmark optimal social welfare is still the one without any additional latencies the coordination mechanism may impose. It has been recently shown that indeed such mechanisms can positively affect social welfare~\cite{christodoulou2011improving}. In this work, the non-decreasing of the latency is respected not for every unit but  on the average.

\section{Model}
We define a congestion game $(E,l, {\cal S}, d)$ in the generic sense but using network routing terminology for convenience. First, a set $E$ of edges with an associated non-negative, non-decreasing, continuously differentiable and  convex $l_e(x_e)$ \emph{latency} function for each edge. We note that these assumptions are typical for latency functions. We have $n$ player types $1,2,\ldots,k$, and for each player type $i$ there is a finite multiset ${\cal S}_i$ of subsets of $E$, called the strategy set of players of type $i$. A particular element $S\in{\cal S}_i$ is a single strategy of player type $i$. We also have a demand $d_i$ for each player type $i$. 

We assume that player types are non-atomic, i.e. they consist of infinitely small users or infinitely divisible traffic. Let $x_i^S$ denote a nonnegative real representing the amount of players of type $i$ that use strategy $S$ and $x_i$ the vector for the strategy set ${\cal S}_i$. The vector $x$ for all $x_i$'s is called a \emph{flow} if for all player types $i$, $\sum_{S\in{\cal S}_i}x_i^S=d_i$. Overloading notation, we define the amount of players of type $i$ having flow through edge $e$, \[ x_e^i=\sum_{\{S:S\in {\cal S}_i,e\in S\}} x_i^S, \] and the total flow through an edge $e$ as\[ x_e=\sum_{i=1\ldots n} x_e^i.\] 

In related literature, the \emph{cost} induced to each player type $i$ by a flow $x$ is the sum of latencies of edges used by all players of the specific type, $c_i(x)=\sum_{e\in E} l_e(x_e)\cdot x_e^i$. The social cost is \[ C(x)=\sum_{e\in E} l_e(x_e)\cdot x_e.\]
For reference, we also note the Wardrop equilibrium in our setting.
\begin{definition}
We say that the flow vector $x$ is in Wardrop equilibrium if for all players' type $i$ and for any pairs of strategies (paths) $S_1, S_2 \in{\cal S}_i$, if $x_i^{S_1} >0$ then the following holds:\begin{equation}
 \sum_{e \in S_1} l_e(x_e) \leq \sum_{e \in S_2} l_e(x_e).  \end{equation} 
\end{definition} 

\section{Variable delay mechanism}
Given a congestion game $(E,l, {\cal S}, d)$ with  non-negative, non-decreasing, continuously differentiable and convex latency functions, we propose a coordination mechanism which induces cost to players  as follows:  
 
Let $N=(N_e)_{e\in E}$ be a sequence of positive integers indexed by the set of elements (edges) $E$  to be called a {\em batch system}. A positive integer $b \leq N_e$ is referred to as a batch index at edge $e$. The integer $N_e$ is the number of batches the flow at each $e$ will be divided.    Specifically at each edge $e$, the flow of any player type through a path $S$ such that $e \in S$, and consequently the total flow $x_e$ through $e$,   are divided into $N_e$ equal batches. Now consider  the following functions, known as marginal-cost latency functions:
\begin{equation}\label{def:delay}
\hat{l}_e(x_e) = l_e(x_e)+l'_e(x_e)\cdot x_e,
\end{equation}
where $l'_e()$ is the derivative of $l_e()$.
The latency induced to players at an edge $e$ is not going to be equal to all. Instead, flow of any player type and through any path $S$ at  batch $b$ receives latency $\hat{l}_e((b/N_e)x_e)$ per unit. Players are interested in minimizing their own latency. Note that this assignment can be performed by any desired policy, for example, randomly, first-come first-serve, by priority lists etc. We will refer to the previous model of applying equal latency to all players as the latencies or original model. 

Since each batch $b$ receives latency $\hat{l}_e((b/N_e)x_e)$ per unit, we  define the  cost with respect to the batch system at an an edge with flow $x_e$  to be:
\begin{equation*}
\hat{c}_e(x_e) = (x_e/N_e) \sum_{b=1}^{N_e} \hat{l}_e((b/{N_e})x_e)
\end{equation*}
and the social cost with respect to the batch system  $$\hat{C}(x) = \sum_e \hat{c}_e(x_e).$$ 
Note that $\hat{c}_e(x_e)\geq \int_0^{x_e} d_e(z) dz=l_e(x_e)x_e = c_e(x_e)$ which is the cost at edge $e$ as before; we do not decrease the cost as per the coordination mechanisms' doctrine and as we shall see later, any cost increase can be made arbitrarily small. 
\begin{definition} \label{def:batchequi} We say that the flow vector $x$ is in equilibrium with respect to the batch system if for all players $i$ and for any pairs of strategies (paths) $S_1, S_2 \in{\cal S}_i$, if $x_i^{S_1} >0$, then  for any sequence of positive integers (batch indices) $(b_e)_{e \in E}$ such that $\forall e \ b_e \leq N_e$, the following holds:\begin{equation}
 \sum_{e \in S_1} \hat{l}_e((b_e/N_e)x_e) \leq \sum_{e \in S_2} \hat{l}_e(x_e).  \end{equation} 
\end{definition} 
Intuitively, if an atom changes path, then we assume that it gets to the last batch of every new edge of the new path, and that under this assumption, there is no strict gain in total per unit latency. Also, at an intuitive level, we assume that the atoms of the flow are indistinguishable, so it makes sense to consider all possible batch assignments along the edges of a path, i.e. even batch assignments where an early batch in an edge in the beginning of a path is followed  by a late batch  in the last edges of the path (if atoms were not indistinguishable, then there might arise issues with the implementation policy on the batch assignments).
We have the following:
\begin{lemma}\label{equiv}
The flow vector $x$ is in equilibrium with respect to the batch system iff it is in Wardrop equilibrium with respect to the marginal-cost latency functions $\hat{l}_e(x_e) = l_e(x_e)+l'_e(x_e)\cdot x_e$, i.e. iff for all players $i$ and for any pairs of strategies (paths) $S_1, S_2 \in{\cal S}_i$, if $x_i^{S_1} >0$, then  $\sum_{e \in S_1} \hat{l}_e(x_e) \leq \sum_{e \in S_2} \hat{l}_e(x_e).$\end{lemma}
\begin{proof}
Since $l_e()$ are convex, $\hat{l}_e$ are non-decreasing. Because the inequality in Definition~\ref{def:batchequi} holds for any selection of batch indices, and therefore also for $b_e =N_e$. The lemma easily follows. \hfill $\Box$
\end{proof}

We now state the following well known facts derived from the literature~\cite{DBLP:books/daglib/0012676}. 
\begin{fact}\label{exist}
When the latencies are non-negative, continuous and non-decreasing, there always exists at least one Wardrop equilibrium.\end{fact}
\begin{fact} If $x$ and $x'$ are flow vectors in  Wardrop equilibrium then $l_e(x_e)x_e=l_e(x'_e)x'_e$ for all edges $e$. This also shows a unique social cost for all Wardrop equilibria.\end{fact}
\begin{fact} \label{opt} A flow vector $x$ in Wardrop equilibrium with respect to the marginal cost latencies 
$l_e(x_e)+l'_e(x_e)\cdot x_e$ has a social cost $C(x)$ with respect to the latency functions $l_e$  which is optimal. \end{fact}

We will now transfer these results in our variable delay batch setting.
\begin{theorem}\label{cumeq}
Under the variable delay mechanism, any batch system has  a uni\-que equilibrium  (as defined in Definition \ref{def:batchequi}). Moreover, 
there is always a suitable batch system whose cost, with respect to the batch system,  when in equilibrium (in the sense of Definition \ref{def:batchequi}) is arbitrarily close to the optimal social cost of the latencies model.
\end{theorem}
\begin{proof} Indeed by the preceding Facts \ref{exist}--\ref{opt}, and by Lemma \ref{equiv}, it suffices to prove that if a flow $x$ is in  equilibrium with respect to a suitable batch system, then its cost $\hat{C}(x)$ with respect to the batch system is  arbitrarily close  to the social cost $C(x)$ of the latencies model. This however is immediate to see since $$\hat{c}_e(x_e) = (x_e/N_e) \sum_{b=1}^{N_e} \hat{l}_e((b/{N_e})x_e)$$ can be made arbitrarily close to $$\int_0^{x_e} \hat{l}_e(z) dz=l_e(x_e)x_e = c_e(x_e)$$ by choosing for each $e$ a large enough $N_e$. \hfill$\Box$
\end{proof}

\section{Discussion}
Aside from the assumptions on the nature of the individual latency functions, perhaps more interestingly, we require a mechanism that can allocate variable cost at the flow within an edge and also ``enforce'' that anyone attempting to switch strategies is ``penalized'' by getting the worst possible cost on his new strategy.

In the context of road networks, these requirements could be implemented in practice by traffic management tools such as traffic lights and access passes. Especially in a dynamic setting one could implement such mechanism by motivating drivers to not deviate from their strategies with incentives (parking, passes, etc.) that vest over time, as long as the driver's strategy remains unchanged.

For applications in the computer or network domain, such a mechanism could be implemented much easier by priorities at the software/hardware level. For example in job scheduling, operating system priority handling could easily be used while in network scheduling, packet switching at the router level could satisfy our requirements.

\section{Acknowledgements}
We would like to thank George Christodoulou for the very helpful discussions.  

\bibliographystyle{plain}
\bibliography{delays}

\end{document}